%% file: kchordal.tex
\documentclass[a4paper, 12pt]{article}
\input{preamble}

\newcommand{\qu}{\textsc{Query}}
\DeclareMathOperator{\tl}{tl}

\title{Tight distance query reconstruction for trees and graphs without long induced cycles}

\author[1]{Paul Bastide}
\author[2]{Carla Groenland}
\affil[1]{LaBRI - Université de Bordeaux, \href{mailto:paul.bastide@ens-rennes.fr}{paul.bastide@ens-rennes.fr}}
\affil[2]{TU Delft, \href{mailto:c.e.groenland@tudelft.nl}{c.e.groenland@tudelft.nl}}

\newcommand{\n}{N} 
\newcommand{\const}{\frac1{200}}
\date{}
\begin{document}

\maketitle
\begin{abstract}
Given access to the vertex set $V$ of a connected graph $G=(V,E)$ and an oracle that given two vertices $u,v\in V$, returns the shortest path distance between $u$ and $v$, how many queries are needed to reconstruct $E$? 

Firstly, we show that randomised algorithms need to use at least $\frac1{200} \Delta n\log_\Delta n$ queries in expectation in order to reconstruct $n$-vertex trees of maximum degree $\Delta$. The best previous lower bound (for graphs of bounded maximum degree) was an information-theoretic lower bound of $\Omega(n\log n/\log \log n)$. Our randomised lower bound is also the first to break through the information-theoretic barrier for related query models including distance queries for phylogenetic trees, membership queries for learning partitions and path queries in directed trees.

Secondly, we provide a simple deterministic algorithm to reconstruct trees using $\Delta n\log_\Delta n+(\Delta+2)n$ distance queries. This proves that our lower bound is optimal up to a multiplicative constant. 
We extend our algorithm to reconstruct graphs without induced cycles of length at least $k$ using $O_{\Delta,k}(n\log n)$ queries. Our lower bound is therefore tight for a wide range of tree-like graphs, such as chordal graphs, permutation graphs and AT-free graphs. The previously best randomised algorithm for chordal graphs used $O_{\Delta}(n\log^2 n)$ queries in expectation, so we improve by a $(\log n)$-factor for this graph class.
\end{abstract}

\section{Introduction}

The distance query model has been introduced \cite{beerliova06}. In this model, only the vertex set $V$ of a hidden graph $G=(V,E)$ is known and the aim is to reconstruct the edge set $E$ via distance queries to an oracle. For a pair of vertices $(u,v) \in V^2$, the oracle answers the shortest path distance between $u$ and $v$ in $G$. The algorithm can select the next query based on the responses of earlier queries. If there is a unique graph consistent with the query responses, the graph has been reconstructed.

For a graph class $\mathcal{G}$ of connected graphs, we say an algorithm reconstructs the graphs in the class if for every graph $G\in \mathcal{G}$ the distance profile obtained from the queries is unique to $G$ within $\mathcal{G}$.  Let $f(G,A)$ denote the number of queries that a deterministic algorithm $A$ takes until it has reconstructed the graph $G\in \mathcal{G}$. The \textit{query complexity} of an algorithm $A$ is the maximum number of queries that $A$ takes on an input graph from $\mathcal{G}$, that is, $\max_{G\in \mathcal{G}}f(G,A)$. A randomised algorithm is a probability distribution $\pi$ over algorithms and its query complexity given by the expected number of queries $\max_{G\in \mathcal{G}}\mathbb{E}_{A\sim \pi}[f(G,A)]$. The (randomised/deterministic) query complexity for reconstructing a graph class is now given by the complexity of the best (randomised/deterministic) algorithm.

Of course, by asking the oracle the distance between every pair $(u,v)$ of vertices in $G$, we can completely reconstruct the edge set as $E = \{ \{u,v\} \mid d(u,v) = 1 \}$. 
This implies a trivial upper bound of $\binom{|V|}2$ on the query complexity. 
Unfortunately, this upper bound can be tight. For example, the clique $K_n$ is difficult to distinguish from $K_n$ minus an edge: if the missing edges is $\{u,v\}$, then the only query answer that is different is then one for the pair $(u,v)$. Thus, all pairs need to be queried before $K_n$ is reconstructed. The core of this problem is in fact high degree vertices \cite{reyzin2007learning} (see \cref{fig:stars}) and therefore we will  restrict our attention to connected $n$-vertex graphs of maximum degree $\Delta$, as has also been done in earlier work.

\begin{figure}[H]
    \centering
    \input{fig/fig_stars}
    \caption{In order to distinguish the tree on the left from all possible labellings of the tree on the right, $\Omega(n^2)$ queries are needed. }
    \label{fig:stars}
\end{figure}

The best known lower bound (for bounded degree graphs) is from \cite{KannanMZ14}: by an information-theoretic argument, $\Omega_{\Delta}(n \log n/\log\log n)$ queries are needed to reconstruct $n$-vertex trees of maximum degree $\Delta$. Let us consider $\Delta = 10$ to sketch the idea of the proof. We use that the class $\mathcal{C}$ of $n$-vertex graphs of maximum degree $10$ and diameter at most $\log n$ has size $\Omega(2^{n\log n})$. For any algorithm distinguishing graphs from $\mathcal{C}$ in $N$ queries, no two distinct graphs from $\mathcal{C}$ can get the same responses to the first $N$ queries. The diameter condition, ensures that (for graphs in  $\mathcal{C}$) every query answer is an element of $\{0,\dots,\lfloor \log n\rfloor \}$ and therefore can be encoded using at most $\log(\log n+1)$ bits. Concatenating the answers to the first $N$ queries, and using that the number of possible strings needs to be at least the number of graphs in the class, we find $2^{N \log(\log n+1)} = \Omega(2^{n\log n})$. This implies $N = \Omega(n \log n / \log \log n)$.

Improving on such an information-theoretic lower bound is often difficult.
More generally, randomised query complexity is infamously difficult to pinpoint: for example, state-of-the-art results are also far from tight bounds for the recursive majority function \cite{goos2016composition,Leonardos,MagniezrecursiveMajority}. In the setting of the evasiveness conjecture, the oracle can answer \textit{adjacency queries} (`given $u,v$, is $\{u,v\}\in E$?') instead of distance queries. It has been shown already in 1975 \cite{rivest1975generalization}
that for any fixed non-trivial monotone graph property of the graph (such as `does $G$ have a triangle') any deterministic algorithm needs $\Theta(n^2)$ on $n$-vertex graphs. At the same time, the best randomised query lower bound  $\Omega(n^{4/3}(\log n)^{1/3})$ from \cite{chakrabarti2001improved}  is far from the best upper bound of $O(n^2)$. Even seemingly simple questions such as estimating the average degree of a graph using vertex degree queries requires new probabilistic tools to achieve tight bounds \cite{feige2004sums,goldreich2008approximating}.  

For $n$-vertex trees of maximum degree $\Delta$, we achieve the correct dependency on $n$ and $\Delta$ for both the randomised and deterministic query complexity using distance queries. Our first result towards this is the following lower bound.
\begin{restatable}{theorem}{lowerboundnlogn}
\label{thm:lowerboundnlogn}
Let $\Delta \geq 2$ and $n=2c \Delta^k$ be integers, where $c\in [1,\Delta)$ and $k\geq 50(c\ln c + 3)$ is an integer. Any randomised algorithm requires at least $\frac1{50} \Delta n\log_{\Delta} n$ queries to reconstruct $n$-vertex trees of maximum degree $\Delta+1$.
\end{restatable}
Note that for any $n\geq 2$, there is a unique $(c,k)\in [1,\Delta)\times \mathbb{Z}_{\geq 0}$ with  $n/2=c\Delta^k$ so the only assumption in our lower bound is that $n$ is sufficiently large compared to $\Delta$. Our result allows $\Delta$ to grow slowly with $n$ (e.g. $\Delta=O((\log n)^{\alpha})$ with $\alpha\in(0,1)$). Moreover, we allow $\Delta$ to be larger for specific values of $n$ (e.g. $O(n^{1/150})$ for $c=1$). We made no attempt to optimize the constant and slightly lowered the constant in the abstract in order to state the result for trees of maximum degree $\Delta$ instead of $\Delta+1$.

This removes the $(1/\log\log n)$ factor compared to the information-theoretic lower bound. Moreover, we achieve both the correct dependence on $n$ and the correct dependency on $\Delta$ (namely $\Delta/\log \Delta$) for the term in front of $n\log n$. This is shown by our following result.
\begin{restatable}{theorem}{treerec}
    \label{thm:tree_rec}
    For all $\Delta\geq 4$, there exists a deterministic algorithm to reconstruct trees of maximum degree at most $\Delta$ on $n$ vertices using $ \Delta n\log_\Delta n +(\Delta +2)n$ queries.
\end{restatable}
For the class of trees of maximum degree $3$, the algorithm from Theorem \ref{thm:tree_rec} for maximum degree $\Delta=4$ is still optimal up to a multiplicative constant (by our Theorem \ref{thm:lowerboundnlogn} applied with $\Delta=3$). 

\cref{thm:lowerboundnlogn} and \cref{thm:tree_rec} show both the  deterministic and randomised distance query complexity of $n$-vertex trees of maximum degree $\Delta$ are $\Theta(\Delta n\log_\Delta n)$ for various ranges of $\Delta$. However, we expect that the randomised complexity will be slightly lower in terms of the multiplicative constant. Towards this, we also show that randomness can be exploited in our algorithm of Theorem \ref{thm:tree_rec} in order to use $\frac12\Delta n\log_\Delta n +(\Delta+1+\log_2n) n$ queries in expectation. 

\medskip

Our algorithm extends to chordal graphs and beyond.
A graph is called \textit{$k$-chordal} if it has no induced cycles of length at least $k+1$. This gives an extension of chordal graphs (which are $3$-chordal graphs). No (randomised nor deterministic) algorithms were previously known with $o(n^{3/2})$ query complexity for $k$-chordal $n$-vertex graphs for $k\geq 5$. 
\begin{restatable}{theorem}{kchordal}
    \label{thm:kchordal} 
    There exists a deterministic algorithm to reconstruct $n$-vertex $k$-chordal graphs of maximum degree at most $\Delta$ using $O_{\Delta,k}(n \log n)$ queries. 
\end{restatable}
Since permutation graphs and AT-free graphs are known to be $5$-chordal and $6$-chordal respectively (see \cite{corneil1997asteroidal,dourisboure2007tree}), our results pinpoint the (randomised and deterministic) query complexity for those graph classes to $\Theta_{\Delta}(n\log n)$.

\paragraph{Previous algorithms and new algorithmic insight}
Kannan, Mathieu, and Zhou \cite{KannanMZ14,mathieu2013graph} designed a randomised algorithm with query complexity $\Tilde{O}_{\Delta}(n^{3/2})$, where the subscript denotes the constant may depend on $\Delta$ and $\Tilde{O}(f(n))$ is a short-cut for $O(f(n)\operatorname{polylog}(n))$. 
In the same article, they give randomised algorithms for chordal graphs and outerplanar graphs with a quasi-linear query complexity $O_{\Delta}(n\log^3 n)$. Rong, Li, Yang, and Wang  \cite{rong21} improved the randomised query complexity for chordal graphs to  $O_{\Delta}(n\log^2 n)$. Their algorithm only requires a weaker type of oracle and applies to graphs without induced cycles of length at least 5. Our algorithm extends the class of graphs and shaves off a $\log n$-factor, achieving the best possible dependence on $n$ by our Theorem \ref{thm:lowerboundnlogn}. Low chordality has also been used in other works for designing efficient algorithms (e.g. routing schemes \cite{dourisboure2005compact}, computing maximal independent sets or maximal packings, etcetera \cite{gartland2021finding,pilipczuk2021quasi}).

Most known algorithms with quasi-linear query complexity in the distance oracle setting are randomised, with a recent work giving a linear deterministic algorithm for interval graphs from Rong, Li, Yang, and Wang  \cite{rong21} as a notable exception.

We provide a new approach for exploiting separators which also extends to various `tree-like' graphs.
Our algorithm restricted to the class of trees is surprisingly simple: we compute a Breadth First Search (BFS) tree starting from a vertex $v_0$ and then inductively reconstruct the tree up to layer $i$. For each vertex in layer $i+1$, we use balanced separators to `binary search' its parent in layer $i$. The algorithm and its analysis is given in \cref{sec:trees} and we extend it to  $k$-chordal in \cref{sec:kchordaldet} using structural graph theory insights.

\paragraph{Lower bound technique}
In order to prove our lower bounds, we first restrict our attention to reconstructing the labelling of the leaves of one particular tree. This reduces some `noise' and reduces the problem to its core.
We prove a lower bound on the number of queries needed to reconstruct this labelling, and so the lower bound holds for any graph class that contains (all labellings of) this particular tree.  

\begin{figure}[H]
    \centering
    \input{fig/introlb}
    \caption{An example of the tree $T_{\Delta,k}$ for $\Delta = 3$ and $k=2$ is depicted with labels. The `$?$'s denote that the labels of the leaves are what needs to be reconstructed.}
    \label{fig:treelowerbound2}
\end{figure}

Our tree of interest is $T_{\Delta,k}$ depicted in \cref{fig:treelowerbound2}. It is the rooted tree of depth $k+1$ where the root (labelled by the empty string) has degree $\Delta$, all vertices at layers $1,\dots,k-1$ have degree $\Delta+1$, and all vertices on layer $k$ have degree 2. Since the structure of the tree is fixed, what remains to be reconstructed is the labelling of the vertices. We moreover fix the labels of all internal nodes, where the label of $v$ incorporates information about the path from $v$ to the root. We prove the following key lemma.

\begin{restatable}{lemma}{tleaflabel}
    \label{lem:t_leaf_label}
    Let  $\Delta\geq 2$, $k\geq 150$ be positive integers. 
    Consider a labelling of the tree $T_{\Delta,k}$ with $\n=\Delta^k$ leaves, where only the labels of the leaves are unknown. Any randomised algorithm requires at least $\frac1{11}\Delta \n \log_\Delta N$  queries in expectation to reconstruct the labelling.
\end{restatable}
\cref{lem:t_leaf_label} above only applies to specific values of $N$ for readability purposes. An extended version is given in \cref{sec:lb} which is needed to obtain a lower bound of all values of $N$ in our applications of the lemma. 

What remains for reconstruction in the lemma is to assign each leaf to its parent (after that, all adjacencies are determined). So the problem above reduces naturally to the problem of reconstructing a bijection $f:[\Delta]^k\to [\Delta]^k$ which is the viewpoint we will take. (We use the short-cut $[m]=\{1,\dots,m\}$.)
To model this, we define two new reconstruction problems and oracles. We present next the simplest one, as we think it is natural and could be studied in its own right.  

The aim is to reconstruct a bijective function $f:[\Delta]^k\to [\Delta]^k$ (where $N=\Delta^k=\Theta(n)$). The reader should see $f$ as the function that maps the label of a leaf to the label of its parent in $T_{\Delta,k}$.
We show that reconstructing $f$ with distance queries is deeply link with reconstructing $f$ with the help of a new oracle, the \emph{coordinate oracle}. This oracle answers the following two types of coordinate queries, where we use $x_i$ to denote the $i^{\text{th}}$ coordinate of $x \in [\Delta]^k$: 
\begin{enumerate}
    \item[(1)] `is $f(a)_i=f(a')_i?$' for $a,a'\in [n]$ and $i\in [k]$, and 
    \item[(2)] `is $f(a)_i=j$?' for $a\in [n],~j\in [\Delta]$ and $i\in [k]$.
\end{enumerate}
Interestingly, instead of the usual number of queries, we can link the complexity of our original problem to the number of \textbf{NO} answers given by the coordinate oracle. The reduction goes via another function reconstruction problem with a more involved type of queries (see Section \ref{sec:lb}).

\paragraph{Applications to other settings} 
The coordinate oracle presented above is of independent interest. Using the key lemma or intermediate results in this new setting, we deduce improved randomised lower bounds for various related reconstruction problems randomised lower bounds for other query models (see Section \ref{subsec:phylogentic} for further details):
\begin{itemize}
    \item betweenness queries in graphs \cite{abrahamsenBodwinRotenbergStockel16} (also called separator queries \cite{TreesSeparatorQueriesJagadishSen}), 
    \item path queries in directed graphs \cite{paralleltreesAfsharea20,AfsharESA20,Wang19},
    \item membership queries for learning a partition \cite{KingZhangZhou,Liu22,scottNIPS}, 
        \item leaf-distance queries in phylogenetic trees \cite{brodal01,hein1989optimal,Kannan96,KaoLingasOstlin99,KingZhangZhou},
    \item comparison queries in tree posets \cite{Roychoudhury23}.
\end{itemize}
Previous work in these settings found deterministic lower bounds or used information-theoretic arguments to obtain weaker randomised lower bounds. Our lower bound often matches the complexity of randomised algorithms in these settings, thereby settling those randomised query complexities (up to a multiplicative constant) as well.

\paragraph{Roadmap} In \cref{sec:prel}, we set up our notation and give the relevant definitions. In  \cref{sec:determinist}, we present our deterministic algorithms for trees and $k$-chordal graphs, obtaining new upper bounds. In \cref{sec:lb}, we prove a matching randomised lower bound and discuss further applications of this lower bound.
In \cref{sec:concl} we conclude with some open problems.

\section{Preliminaries}
\label{sec:prel}
In this paper, all graphs are simple, undirected and connected except when stated otherwise. All logarithms in this paper are base 2, unless mentioned otherwise, where $\ln=\log_e$. For $a \leq b$ two integers, let $[a,b]$ denote the set of all integers $x$ satisfying $a \leq x \leq b$. We short-cut $[a]=[1,a]$.

For a graph $G$ and two vertices $a,b \in V(G)$, we denote by $d_G(a,b)$ the length of a shortest path between $a$ and $b$. For $G = (V,E)$, $A \subseteq V$ and $i \in \NN$, we denote by $N^{\leq i}_G[A] = \{v \in V \mid  \exists a \in A, d_G(v,a) \leq i\}$. We may omit the superscript when $i=1$. We write $N_G(A)=N_G[A]\setminus A$ and use the shortcuts $N_G[u],N_G(u)$ for  $N_G[\{u\}],N_G(\{u\})$ when $u$ is a single vertex.  We may omit the subscript when the graph is clear from the context. 

\paragraph{Distance queries}
We denote by $\qu_G(u,v)$ the call to an oracle that answers $d_G(u,v)$, the distance between $u$ and $v$ in a graph $G$. For two sets $A,B$ of vertices, we denote by $\qu_G(A,B)$ the $|A|\cdot|B|$ calls to an oracle, answering the list of distances $d_G(a,b)$ for all $a \in A$ and all $b \in B$. We may abuse notation and write $\qu_G(u,A)$ for $\qu_G(\{u\},A)$ and may omit $G$ when the graph is clear from the context.

For a graph class $\mathcal{G}$ of connected graphs, we say an algorithm reconstructs the graphs in the class if for every graph $G\in \mathcal{G}$ the distance profile obtained from the queries does not belongs to any other graph from $\mathcal{G}$. The \textit{query complexity} is the maximum number of queries that the algorithm takes on an input graph from $\mathcal{G}$, where the queries are adaptive. For a randomised algorithm, the query complexity is given by the expected number of queries (with respect to the randomness in the algorithm).

\paragraph{Tree decomposition}
A \textit{tree decomposition} of a graph $G$, introduced by  \cite{robertson1986graph}, is a tuple $(T,(B_t)_{t \in V(T)})$ where $T$ is a tree and $B_t$ is a subset of $V(G)$ for every $t \in V(T)$, for which the following conditions hold.
    \begin{itemize}
        \item For every $v \in V(G)$, the set of $t \in V(T)$ such that $v \in B_t$, is non-empty and induces a subtree of $T$.
        \item For every $uv \in E(G)$,
        there exists a $t \in V(T)$ such that $\{u,v\} \subseteq B_t$.
    \end{itemize}

\paragraph{Balanced separators} For $\beta\in (0,1)$, a \textit{$\beta$-balanced separator} of a graph $G = (V,E)$ for a vertex set $A\subseteq V$ is a set $S$ of vertices such that the connected components of $G[A \setminus S]$ are of size at most $\beta |A|$. 

One nice property of tree decompositions is that they yield $\frac12$-balanced separators.
\begin{lemma}[\cite{robertson1986graph}]
    \label{lem:td_sep}
    Let $G$ be a graph, $A\subseteq V(G)$ and $(T,(B_t)_{t \in V(T)})$ a tree decomposition of $G$. Then there exists $u \in V(T)$ such that the bag $B_u$ is a $\frac12$-balanced separator of $A$ in $G$. 
\end{lemma}
In fact,  for trees (of treewidth $1$), we can always find a $\frac12$-balanced separator of size $1$ (that is, a single vertex). We refer to such a separator as a \textit{vertex separator} or \textit{separating vertex}.

\section{Improved algorithm for distance reconstruction}
\label{sec:determinist}
This section presents our deterministic algorithms for the class of trees and $k$-chordal graphs. The lower bounds are given in \cref{sec:lb}.

\subsection{Optimal algorithm for tree reconstruction}
\label{sec:trees}

We first describe our algorithm for reconstructing trees as it encapsulates most of the algorithmic ideas, while being fairly simple.

\treerec*

\begin{proof}
    Let $T$ be a tree on $n$ vertices, and let $\Delta$ be the maximum degree of $T$. Our algorithm starts as follows. We  pick an arbitrary vertex $v_0\in V(T)$ and will consider (for the analysis) the input tree $T$ as rooted in $v_0$. We call $\qu(v_0,V(T))$. We define the $i$\textsuperscript{th} layer of $T$ as $L_i = \{v \in V(T) \mid d(v_0,v) = i \}$. We proceed to reconstruct the graph induced by the first $i$ layers by induction on $i$. 
    
    Note that $T[L_0] = (\{v_0\},\emptyset)$ is immediately reconstructed. 
    We fix an integer $i \geq 1$ and assume that the first $i-1$ layers are fully reconstructed (i.e we discovered all the edges and non-edges of $T[L_0 \cup \cdots \cup L_{i-1}]$). Let $T' = T[L_0 \cup \cdots \cup L_{i-1}]$ be the already reconstructed subtree. We show how to reconstruct the edges between the $(i-1)$\textsuperscript{th} layer and the $i$\textsuperscript{th} layer. Note that this  suffices to reconstructs all the edges (since in a tree, edges can only be between consecutive layers).
    
    Choose an arbitrary vertex $v \in L_i$. We first show that we can find the parent of $v$ in $L_{i-1}$ using $O(\Delta \log n)$ queries and then describe how to shave off an additional $(\log \Delta)$ factor.
    
    The procedure goes as follows. Let $T_1=T'$. As $T$ is a tree, it admits a $\frac12$-balanced separator of size 1. Let $s_1$ be a vertex for which $\{s_1\}$ forms such a separator. We ask first $\qu(v,N[s_1])$, where the neighbourhood is taken in $T_1$. As $T$ is a tree, there is a unique path between any two vertices. So for $w \in N(s_1)$, the distance $d(v,w)= d(v,s_1)-1$ if $w$ lies on the shortest path from $v$ to $s_1$ and $d(v,s_1)+1$ otherwise. From this, we can infer the neighbour $x$ of $s_1$ that is the closest to $v$ as the one for which the answer is smallest (or find that $s_1$ is adjacent to $v$ and finish). Moreover, the unique path from $s_1$ to $v$ lives in the connected component $T_2$ of $T_1 \setminus \{s_1\}$ that contains $x$. In particular, 
    $T_2$ contains the parent of $v$ (see \cref{fig:trees}).
    
    \begin{figure}[ht]
        \centering
        \input{fig/fig_trees}
         \caption{The subtree $T_2$ contains the neighbour of $v$ on a shortest path to $s_1$ and so contains the parent of $v$.}
        \label{fig:trees}
    \end{figure}
    
    We can repeat this process and construct two sequences $(T_j)_{j\in \mathbb{N}}$ and $(s_j)_{j\in \mathbb{N}}$, where $T_{j+1}$ is the connected component of $T_j \setminus \{s_j\}$ containing the parent of $v$ and $s_j\in V(T_j)$ is chosen so that $\{s_j\}$  is a $\frac12$-balanced separator of $T_j$. 
    Once $T_\ell$ contains less than $\Delta +1$ vertices for some $\ell$ or the vertex $s_\ell$ is identified as the parent of $v$, we finish the process\footnote{If desired, we may define $T_{j}=T_\ell$ and $s_j=s_\ell$ for all $j\geq \ell$.}.
    By definition of $\frac12$-balanced separator,
    \[
    \forall j \in [\ell-1], |T_{j+1}| \leq |T_{j}|/2 \text{\quad and thus \quad} \ell \leq \log n. 
    \]
    If the process finished because $T_\ell$ has at most $\Delta$ vertices, we use at most $\Delta$ additional queries via $\qu(T_\ell,v)$. We infer the parent of $v$ from the result. 
    For each $j\leq \ell$, we use at most $\Delta+1$ queries to reconstruct $T_{j+1}$ from $T_j$. Hence we use $O(\Delta \log n)$ queries in total.
    
    Taking a closer look at the process, at any step $j$, we can choose the order on the queries $\qu(v,w)$ for $w\in N(s_j)$ and may not need to perform all the queries. 
    Given a subtree $S$ of $T'$ on $b\geq 1$ vertices that contains the parent of $v$, we now show how to find the parent of $v$ in $f(b)=\Delta \log_\Delta b+\Delta+1$ queries (giving the desired improvement of a $(\log \Delta)$ factor). If $S$ has at most $\Delta+1$ vertices, we may simply $\qu(v,S)$ and deduce the answer. Otherwise, let $s\in S$ be a $\frac12$-balanced separator for $S$. This has at least two neighbours since $S$ has at least $\Delta+1$ vertices. We order the connected components of $S \setminus \{s\}$ by non-increasing size, and ask the queries in the same order:  we start with $\qu(v,w_1)$ for $w_1$ the neighbour of $s$ which is in the largest component, then proceed to the neighbour of the second largest component etcetera. We terminate when we find two different distances or queried all the neighbours. In particular, we never perform $\qu(v,s)$.
    \begin{itemize}
        \item If $d(v,w)$ for $w\in N(s)$ are all the same then $s$ is the parent of $v$. We terminate and recognise $s$ as the parent of $v$. We used at most $\Delta\leq f(b)$ queries.
        \item If we discover that $d(v,w)< d(v,w')$ for some $w,w'\in N(s)$, then $s$ is not the parent of $v$. In fact, $w$ is the vertex from $N[s]$ closest to $v$ and we recursively perform the same procedure to the subtree $S'$ of $S\setminus\{s\}$ that contains $w$. Note that $S'$ must contain the parent of $v$.
    \end{itemize}    
    If we query $2$ neighbours of $s$ before detecting the component containing the parent of $v$, our next subtree $S'$ satisfies $|S'| \leq |S|/2$ since $\{s\}$ is a $\frac12$-balanced separator. If we query $m\geq 3$ neighbours of $s$ before detecting the component containing the parent of $v$, our next subtree $S'$ satisfies $|S'| \leq |S|/m$ since there are $m-1$ components of $S\setminus\{s\}$ that are at least as large. Either way, we decrease the size of the tree by a factor at least $x$ if we perform $x$ queries, where $x\in \{2,\dots,\Delta\}$. 

    We show by induction on $b$ that the procedure described above uses at most $f(b) = \Delta \log_\Delta b+\Delta+1$ queries. 
The claim is true when $b\leq \Delta+1$. By the discussion above, for $b\geq \Delta+2$, the process either finishes in $\Delta$ queries or uses $x+f(b')$ queries for some $b'\leq b/x$ and $x\in \{2,\dots,\Delta\}$. It thus suffices to show that 
    \[
    f(b/x)+x\leq f(b) \text{ for all $x\in \{2,\dots,\Delta\}$}.
    \]
    By definition, $f(b)-f(b/x)=\Delta \log_\Delta x$. We show that $\Delta \log_\Delta x\geq x$ for all $x\in [2,\Delta]$. By analysing the derivative of $\Delta \log_\Delta x -x$ on the (real) interval $x \in [2,\Delta]$, we find that the minimum is achieved at $x=2$ or $x=\Delta$. Since $\Delta\geq 4$, the minimum is achieved at $x=\Delta$, as desired. 
    
    With the improved procedure, we can reconstruct the edge from $L_{j-1}$ to $v$ in at most $\Delta \log_\Delta n+\Delta+1$ queries. Repeating the same strategy to reconstruct the parent of every vertex, we obtain the edge set of $T$ in at most
    \[
    (n-1)+(n-1)(\Delta \log_\Delta n +(\Delta +1))\leq      \Delta n\log_\Delta n +(\Delta +2)n
    \]
    queries.
\end{proof}
Even though we show in the next section that we cannot achieve a better dependency in $(n,\Delta)$ using randomisation, we can improve the multiplicative constant.
\begin{theorem}
    \label{thm:rand_alg}
    For any $\Delta \geq 4$, there exists a randomised algorithm for reconstructing $n$-vertex graphs of maximum degree at most  $\Delta$ using $(\frac12\Delta/\log(\Delta)+1)\log n+(\Delta+2) n$ queries in expectation.
\end{theorem}
\begin{proof}
The algorithm works similarly to the algorithm from \cref{thm:tree_rec}.
We define the same layers and inductively reconstruct the graph induced on the first $i$ layers. We find the parent of a vertex $v\in L_i$ via a similar sequence of separators $s_1,\dots,s_j$ and trees $T_1\supseteq T_2 \supseteq \dots\supseteq T_j$.
The key difference is that when we wish to learn the vertex in $N[s_j]$ closest to $v$, then we perform $\qu(v,w)$ for $w\in N[s_j]$ in an order that is chosen at random.
Suppose that $|T_j|=b$.

\begin{claim}
Let $T$ be a tree and let $t\in V(T)$.
Let $a_1,\dots,a_k$ be the sizes of the components of $T \setminus \{t\}$ and let $v_1,\dots,v_k$ denote the neighbours of $s_j$ in these components. There is a random order on $v_1,\dots,v_k$ such that the expected number of vertices placed before $v_i$ is at most $\tfrac12 \tfrac1{a_i} \sum_{j=1}^ka_j$ for all $i\in [k]$.
\end{claim}
\begin{proof}
We generate the order by  independently sampling $X_i\sim U[0,a_i]$ uniformly at random for all $i\in [k]$, where $[0,a_i]$ denotes the set of real numbers between $0$ and $a_i$. Almost surely, $X_{\pi(1)}>\dots>X_{\pi(k)}$ for some permutation $\pi$ on the support $[k]$ and this gives us our desired random order. 
\renewcommand{\E}{\mathbb{E}}
\renewcommand{\P}{\mathbb{P}}

We prove that the expected number of vertices placed before $v_1$ is at most $\frac12\frac{a_2+\dots+a_k}{a_1}$ and then the remaining cases will follow by symmetry.
Let $I(x_1)$ denote the number of vertices placed before $v_1$ given that $X_1=x_1$, i.e. the number of $i\in \{2,\dots,k\}$ such that $X_i>x_1$:
\[
I(x_1)=\sum_{i=2}^k \text{Bern}\left(\max\left(\frac{a_i-x_1}{a_i},0\right)\right).
\]
A Bernoulli random variable with probability $p$ has expectation $p$. 
The expected number of vertices placed before $v_1$ is hence
\[
\frac1{a_1}\int_0^{a_1} \E[I(x_1)]dx_1=\sum_{i=2}^k \frac1{a_1}\int_0^{\text{min}(a_1,a_i)}1-\frac{x_1}{a_i}dx_1.
\]
We now show for all $i\in [2,k]$ that the $i$th summand is at most $\frac12 \frac{a_i}{a_1}$, which implies that the number of vertices placed before $v_1$ is indeed at most $\frac12 \frac{a_2+\dots+a_k}{a_1}$. We compute
\[
\frac1{a_1}\int_0^{\text{min}(a_1,a_i)}1-\frac{x_1}{a_i}dx_1=
\frac{\min(a_1,a_i)}{a_1}\left(1 - \frac{\min(a_1,a_i)}{2a_i}\right).
\]
When $a_i\leq a_1$, the expression simplifies to $\frac{a_i}{a_1}\frac12$ as desired. 
When $a_i\geq a_1$, the expression simplifies to $1-\frac12 \frac{a_1}{a_i}$ 
which is at most $\frac12 \frac{a_i}{a_1}$ since 
\[
a_1a_i\leq \frac12 a_i^2+\frac12 a_1^2.\qedhere
\]
\end{proof}
By the claim, if the parent of $v$ is in a component $T_{j+1}$ of $T_j-s_j$ of size $a$, then we query at most $\frac12(b-a)/a$ vertices in expectation before we query the neighbour of $s_j$ in $T_{j+1}$. This means that for a size reduction of $x=b/a$, we perform approximately $\frac12x$ queries in expectation, compared to $x$ in our deterministic algorithm. 
Using linearity of expectation, we will repeat a similar calculation to the one done in the proof of \cref{thm:tree_rec}. 

We show that we reconstruct the edge to the parent of $v$ using at most $\frac12\Delta \log_\Delta n+\Delta+1+\log_2n$ queries in expectation. We show that once we have identified a subtree $T_j$ containing the parent of $v$ with $|T_j|=b$, we use at most $\Delta \log_\Delta b+\Delta+1+\log_2b$ queries in expectation to find the parent of $v$. Let $s_j$ be a $1/2$-balanced separator of $T_j$. The first base case is given when $s_j$ is the parent of $v$: in this case we perform at most $\Delta$ queries (all neighbours of $s_j$) and find that $s_j$ is the parent. The second base case is when $|T_j|=\Delta+1$, in which case we query all distances between $v$ and $T_j$ and identify the parent of $v$.

We now assume that $b\geq \Delta+2$ and that the parent of $v$ is in a component $T_{j+1}$ of $T_j\setminus s_j$. Let $a=|T_{j+1}|$. By the claim, we query at most $\frac12(b-a)/a$ vertices in expectation before we query the neighbour $w$ of $s_j$ in $T_{j+1}$. We still need to query the distance from $v$ to $w$. If $w$ is the first vertex to be queried, then we need to query one more vertex. Since $s_j$ is a $1/2$-balanced separator, this happens with probability at most $1/2$. So after at most $\frac12 \frac{b}a+1$ in expectation we find the neighbour of $v$ is in $T_{j+1}$ after which by induction we need another $\frac12\Delta \log_\Delta a+\Delta+1+\log_2a$ queries in expectation. In total we use at most
\[
\frac12 \frac{b}a+1+\frac12\Delta \log_\Delta a+\Delta+1+\log_2a
\]
in expectation. We find that $\frac{b}a+\Delta \log_\Delta a\leq \Delta\log_\Delta b$ for all $\Delta \geq 4$ (same calculation as before) and $\log_2a+1\leq \log_2b$ since $a\leq b/2$. So we used at most $\Delta \log_\Delta b+\Delta+1+\log_2b$ queries in expectation, as claimed.
\end{proof}

\subsection{Optimal algorithm for $k$-chordal graphs}
\label{sec:kchordaldet}

Using additional structural analysis, we extend our algorithm from trees to $k$-chordal graphs: graphs without induced cycles of length at least $k+1$. 
In the simpler case of (3-)chordal graphs, (randomised) reconstruction from a quasi-linear number of queries was already known to be possible. Besides extending the class of graphs, our algorithm shaves off a $(\log n)$ factor and is now optimal in $n$ (the number of vertices of the input graph). 

The core of the proof uses the same principles as for trees in \cref{thm:tree_rec}: we reconstruct the edges of a vertex $u$ to the previous layer, layer-by-layer and vertex-by-vertex. The two important ingredients are (1) a structural result on the neighbourhood of a vertex (see \cref{cl:kchordal_neighbour}) and  (2) the existence of `nice' balanced separators on the already reconstructed subgraph (see \cref{cl:quality_sep}). After removing the separator, we need to show that we can correctly determine the component that contains the neighbourhood of the vertex $u$ that we are currently considering. We also need to reconstruct the edges within the layer, but this turns out to be relatively easy.

\kchordal*

During the proof we will need the following definition. The \textit{treelength} of a graph $G$ (denoted $\tl(G)$) is the minimal integer $\ell$ for which there exists a tree decomposition $(T,(B_t)_{t\in V(T)})$ of $G$ such that $d_G(u,v) \leq \ell$ for every pair of vertices $u,v$ that share a bag (i.e. $u,v\in B_t$ for some $t\in V(T)$). In particular we will use the following result proved by Kosowski, Li, Nisse, and Suchan \cite{kosowski2015k}.

\begin{lemma}[\cite{kosowski2015k}]
    \label{lem:kchordaltreelengthk}
    For any $k \in \mathbb{N}$, any $k$-chordal graph has treelength at most $k$.
\end{lemma}

\begin{proof}[Proof of \cref{thm:kchordal}]
We start by fixing a vertex $v_0$ and asking $\qu(V(G),v_0)$. From that, we reconstruct $L_i = \{v \in V(G) \mid d(v,v_0) = i\}$. We write $L_{\bowtie i} = \cup_{j \bowtie i} L_j$ for any relation $\bowtie \, \in \{\leq,<,>,\geq\}$.

The algorithm proceeds by iteratively reconstructing  $G[L_{\leq i}]$ for increasing values of $i$. Note that we can reconstruct $L_{\leq 2\Delta k}$, the vertices at distance at most $2\Delta k$ from $v_0$, using $O_{k,\Delta}(n)$ queries. 

Suppose that we reconstructed $G_1 := G[L_{\leq i-1}]$ for some $i\geq 2\Delta k$ and we again want to reconstruct the two edge sets
\[E_{i-1,i} = \{uv \in E(G) \mid u \in L_{i-1}, \, v \in L_{i}\}\] and \[E_{i,i} = \{uv \in E(G) \mid u , v \in L_{i}\}.\]
We call $H_1=G[L_{\leq i-1-k}]$ the core of $G_1$.
We need a lemma that implies that neighbourhoods are not spread out too much in $G_1$.
\begin{claim}
\label{cl:kchordal_neighbour}
For all $u \in L_i$ and $v,w \in N(u) \cap L_{i-1}$, $d_{G_1}(v,w) \leq \Delta k$.
\end{claim}
\begin{proof}
Let $v,w \in N(u) \cap L_{i-1}$ and let $P$ be a shortest $vw$-path in $G_1$. If $V(P) \cap N(u) = \{v,w\}$, then the vertex set $V(P) \cup \{u\}$ induces a cycle in $G$, and so $|V(P)| \leq k$ (else the $k$-chordality would be contradicted). 
For the same reason, $P$ can have at most $k-1$ consecutive vertices outside of $N(u)$. Since $u$ has at most $\Delta$ neighbours, it follows that $d_{G_1}(v,w)\leq \Delta k$.
\end{proof}
Since $G$ has treelength at most $k$ by \cref{lem:kchordaltreelengthk}, it has a tree decomposition $(T',\mathcal{B}')$ such that all bags $B'\in \mathcal{B}'$ satisfy $d_{G}(u,v)\leq k$ for all $u,v\in B'$. In particular the bags have size at most $\Delta^k+1$. 

We have already reconstructed $G_1$, so in particular, we know $N^{\leq k}[v]$ for all $v\in H_1$.
Therefore, we can construct a tree decomposition $(T,\mathcal{B})$ of $G_1$ such that each $B \in \mathcal{B}$ has size at most $\Delta^k+1$ and for any bag $B \in \mathcal{B}$ that contains at least one vertex of the core $H_1$, we have $d_{G_1}(u,v)\leq k$ for all $u,v\in B$.

Fix $u\in L_i$. We describe an algorithm to reconstruct $N(u) \cap L_{i-1}$. The algorithm recursively constructs a sequence of connected graphs $(G_j)_{j=0}^\ell$ and a sequence of separators $(S_j)^\ell_{j=1}$ for some $\ell\leq \lceil \log(n) \rceil$, such that $S_j$ is a $\frac12$-balanced separator of $G_j$, $S_j\subseteq L_{\leq i-\Delta k - 1}$ and $N(u) \cap L_{i-1}\subseteq V(G_j)$. 

We first prove the following claim that we use to find our sequence of separators $(S_j)$. 
\begin{claim}
    \label{cl:quality_sep}
    For $n$ large enough compared to $\Delta$ and $k$ and any set of vertices $A \subseteq V(G_1)$ with $|A| \geq \log n$, there exists a bag $B$ of $T$ such that $B$ is a $\frac12$-balanced separator of $A$, and $B$ is contained in $L_{\leq i-\Delta k - 1}$.   
\end{claim}

\begin{proof}
    Let $T$ be rooted in a bag that contains $v_0$.
    By \cref{lem:td_sep}, there is a bag $B$ of $T$ that forms a $\frac12$-balanced separator for $A$ (i.e. all connected components of $G_1[A\setminus B]$ are of size at most $|A|/2$). We choose such a bag $B$ of minimum depth (in $T$). We need to show $B$ is contained in $L_{\leq i-\Delta k - 1}$.   
    
    If $B$ contains $v_0$, then $B \subseteq G[L_{\leq k}]$. Since $i\geq 2\Delta k$, we are done in this case.
    
    Suppose now that $v_0 \notin B$ and let $B'$ be the parent of $B$. By definition, $B'$ is not a $\frac12$-balanced separator of $A$.     
    If $B$ contains a vertex of $L_{\leq i-1-k}=V(H_1)$, then its diameter in $G_1$ is at most $k$. So either $B\subseteq L_{\leq i-\Delta k-1}$ or $B\subseteq L_{>i-(\Delta+1)k-1}$. We are done in the first case, so assume the latter. Since $G_1$ is connected, $B\cap B'\neq \emptyset$. The same diameter argument gives that $B'\subseteq L_{>i-(\Delta+2)k-1}$. If $C$ is a component of $G_1\setminus B'$ that does not contain $v_0$, then the shortest path (of length at most $i$) from any $v\in C$ must go through $B'$ (at distance at least $i-(\Delta+2)^k-1$ from $v_0$). In particular, all such components are contained in $N^{((\Delta+2)^k+1)}(B')$ and so the total size is at most $\Delta^{(\Delta+2)k + 1} |B'| =O_{k,\Delta}(1)$. For $n$ sufficiently large, this is at most $\frac12\log n$.

    On the other hand, as $B'$ is not a $\frac12$-balanced separator, there exist a component $A'$ of $G_1[A\setminus B']$ with $|A'| > \frac12 \log n$. We found above that $A'$ must contain $v_0$. Since $B$ does not contain $v_0$, $A'$ is contained in the component of $G_1[A \setminus B]$ containing $v_0$. This yields a contradiction with the fact that $B$ is a $\frac12$-balanced separator of $A$. 
\end{proof}

Suppose that we have defined $G_j$ for some $j \geq 1$ and let us describe how to define $G_{j+1}$. If $|G_j| \leq \log n$, we ask $\qu(u,V(G_j))$ and output $N(u)\cap V(G_j)$.
Otherwise, we let $S_j$ be the bag found in \cref{cl:quality_sep} when applied to $A=V(G_j)$. Then $S_j \subseteq L_{\leq i -\Delta k - 1}$  and it is a $\frac12$-balanced separator of $G_j$. Since it is a bag of $T$ and contained in $H_1$, we find the size of the bag is at most $\Delta^k+1$ and $d_{G_1}(u,v)\leq k$ for all $u,v\in S_j$. We ask $\qu(N[S_j],u)$ and let $G_{j+1}$ be a component of $G_j\setminus S_j$ that contains a vertex from $\arg\,\min_{x\in N[S_j]} d_G(x,u)$. 
Then, we increase $j$ by one and repeat the same procedure.

\vspace{0.4cm}

We now prove the correctness of the algorithm presented above.

We first argue that $N(u)\cap L_{i-1}$ is contained in a unique connected component of $G_j \setminus S_j$. Since every separator is included in $L_{\leq i-\Delta k-1}$, we find $d_{G_1}(u,S_\ell) \geq \Delta k + 1$ for all $\ell \leq j$. By \cref{cl:kchordal_neighbour}, all vertices in $N(u)\cap L_{i-1}$ are connected via paths in $G_1$ of length at most $\Delta k$ and by the observation above, these paths avoid all separators so will be present in a single connected component of $G_j\setminus S_j$.
The only thing left to prove is the following claim.
\begin{claim}
    \label{cl:kchordal_induction}
    $G_{j+1}$ is the component that contains  $N(u)\cap L_{i-1}$.
\end{claim}

\begin{figure}
    \centering
    \input{fig/kchordal_fig}
    \caption{This figure depicts a possible configuration in the proof of \cref{cl:kchordal_induction} for which we end up with a contradiction by finding a large induced cycle.}.
    \label{fig:kchordal_fig}
\end{figure}

\begin{proof}
Let $x \in N[S_j]$ such that $d_G(x,u)$ is minimised and let $H_x$ be the connected component of $x$. We will prove that $H_x$ contains $N(u)\cap L_{i-1}$. In particular, this implies that $G_{j+1}=H_x$ (a priori, $G_{j+1}$ could be a different component for another minimiser than $x$), so that $G_{j+1}$ contains $N(u)\cap L_{i-1}$, as desired. 

Suppose towards a contradiction that $N(u) \cap L_{i-1}$ is instead contained in a different connected component $H$. We will find an induced cycle of length at least $k+1$. 

Let $P$ be a shortest path in $G$ from $x$ to $u$. 

Let $P'$ be a path from $u$ to $S_j$ with all internal vertices in $H$.  Such a path exists since $G_j$ is connected.

Let $P''$ be a shortest path in $G$ between a neighbour of $x$ in $S_j$ to the endpoint of $P'$ in $S_j$.
As $S_j$ is a bag contained in $H_1$, any two vertices in $S_j$ are within distance $k$ in $G_1$. So $P''\subseteq L_{\leq i-2k-2}$ (we may assume $\Delta\geq 4$). 

Let $y$ be the first vertex on the path $P$ (from $x$ to $u$) that lies in $L_{i}$ (such a vertex must exist since the path does not have internal vertices in $S_j$ by choice of $x$ and since $H_x$ contains no neighbours of $u$). 

Let $x_1,\dots,x_k$ be the $k$ vertices before $y$ in $P$. Note that none of the $x_i$ can be adjacent to or part of $P'\cup P''$ (since they are in $H_x\cap L_{\geq i-k}$). Let $G'$ be the graph obtained from $G[P\cup P''\cup P']$ by contracting  $P''\cup P'\cup (P\setminus \{x_{1},\dots,x_{k}\})$ to a single vertex $p$. Note that the selected vertex set is indeed connected and that the resulting graph has vertex set $\{x_{1},\dots,x_{k},p\}$. Since $P$ was a shortest path in $G$, the vertex set $\{x_{1},\dots,x_{k}\}$ still induces a path and it suffices to argue about the adjacencies of $p$. Via edges of $P$, the vertex $p$ is adjacent to $x_{1}$ and $x_{k}$. If $p$ was adjacent to $x_i$ for some $i\in [2,k-1]$, then there must be a vertex $y\in P''\cup P'\cup (P\setminus \{x_{1},\dots,x_{k}\})$ adjacent to $x_i$. But a case analysis shows this is not possible. (The only vertices adjacent to $x_i$ in $P$ are $x_{i+1}$ and $x_{i-1}$ since $P$ is a shortest path; we already argued that $y\not\in P'\cup P''$.) We hence found an induced cycle of length $k+1$, a contradiction.
\end{proof}
We now show how to reconstruct all edges in $E_{i,i}$ incident to $u$.
\begin{claim}
\label{cl:kchordal_2neighbour}
    If $x \in N(u)\cap L_{i-1}$ and $y\in N(v)\cap L_{i-1}$ for some $uv \in E_{i,i}$,  then $d_{G[L_{\leq i-1}]}(x,y) \leq 2\Delta k$.
\end{claim}
\begin{proof}
   Since $|N[\{u,v\}]| \leq 2\Delta$, it suffices to prove that $d_{G_1}(x,y)\leq k$ when the shortest path $P$ in $G_1$ between $x$ and $y$ avoids other vertices from $N[\{u,v\}]$. As we argued in \cref{cl:kchordal_neighbour} this is true when $x$ or $y$ is a neighbour of both $u$ and $v$ (else we create an induced cycle of length at least $k+1$). So we may assume that $x\in N(u)\setminus N(v)$ and $y\in N(v)\setminus N(u)$. But now $P \cup \{u,v\}$ is an induced cycle of length at least $k+1$.
\end{proof}
Let $G_1'$ be the graph obtained from $G_1$ by adding the vertices in $L_i$ and the edges in  $E_{i,i-1}$.  Our algorithm already reconstructed $G_1'$. If $uv\in E_{i,i}$, then applying \cref{cl:kchordal_2neighbour} to vertices $x,y\in L_{i-1}$ on the shortest paths from $u,v$ to the root $v_0$ respectively, we find that $d_{G_1'}(u,v)\leq 2\Delta k+2$.
For each $u\in L_i$, we ask $\qu(u,N_{G'_1}^{\leq 2\Delta k +2}(u) \cap L_i)$
 and we record the vertices $v$ for which the response is 1. Those are exactly the vertices adjacent to $u$.
Per vertex $u\in L_i$, this takes at most $\Delta^{2\Delta k +3}$ queries.

The query complexity of reconstructing $E_{i-1,i}$ is $O_{k,\Delta}(\log n|L_i|)$ as there are at most $\log n$ iterations (using the fact that the $(S_j)_j$ are $\frac12$-balanced separators) and in each iteration we do $O_{k,\Delta}(\log n)$ queries per vertex $u\in L_i$. In order to reconstruct $E_{i,i}$, we use $O_{k,\Delta}(1)$ queries per vertex of $L_i$. Therefore, the total query complexity of the algorithm is $\sum_i O_{k,\Delta}(|L_i|\log n) = O_{k,\Delta}(n\log n)$.
\end{proof}


\section{Lower bounds for randomised tree reconstruction}
\label{sec:lb}

In this section we show that the algorithm presented in the previous section is optimal in terms of the dependency on $n$ and $\Delta$, even when randomisation is allowed. 


\lowerboundnlogn*

Note that, for constant $c$, $\Delta$ could even be a small polynomial in $n$.
Any `algorithm' is allowed to be randomised unless specified to be deterministic. 

\subsection{Reconstructing functions from the coordinate oracle} 
\label{sec:coordoracle}
In order to prove the lower bound we reduce to a natural function reconstruction problem that could be of independent interest.  Let $\Delta \geq 3$, $k\geq 1$ and  $n=c\Delta^k$ be integers, where $c\in [1,\Delta)$.

Let $A = [n]$ and $B = [\Delta]^k$. Suppose that $f: A\to B$ is an unknown function that we want to reconstruct. For $b\in B$ and $1\leq i \leq k$, we write $b_i$ for the value of the $i$th coordinate of $b$.

The \textit{coordinate oracle} can answer the following two types of queries:
\begin{itemize}
    \item \textbf{Type 1}. $\qu_1^c(a,b,i)$ for $a \in A$, $b \in [\Delta]$ and $i\in [k]$ answers \textbf{YES} if $f(a)_{i} = b$ and \textbf{NO} otherwise.
    \item \textbf{Type 2}. $\qu_2^c(a,a',i)$ for $a,a' \in A$  and $i\in [k]$ answers \textbf{YES} if $f(a)_{i} = f(a')_i$ and \textbf{NO} otherwise.
\end{itemize}
In the case of the coordinate oracle, we will count the number of queries for which the answer is \textbf{NO} instead of the number of queries. 

We say that $f : A \to B$ is a \emph{balanced function} if for every $b\in B, \, |f^{-1}(b)| = c$ for some integer $c\geq 1$.  

Our main result on function reconstruction from a coordinate oracle is the following. 

\begin{restatable}{theorem}{lbcoord}
    \label{thm:lb_coord}
    Let $\Delta \geq 3$, $c\leq  \Delta-1$ and $k \geq 50(c\ln c +2)$ be positive integers and let $n=c\Delta^k$. Any algorithm reconstructing $f:[n] \to [\Delta]^k$ using the coordinate oracle, in the special case where $f$ is known to be a balanced function, has at least $\frac1{11} \Delta n k$ queries answered \textup{\textbf{NO}} in expectation. 
\end{restatable}

In order to prove \cref{thm:lb_coord}, we first study the  query complexity in the general case, when no restriction is put on $f$. Using Yao's minimax principle \cite{yao1977probabilistic}, studying the expected complexity of a randomised algorithm can be reduced to studying the query complexity of a deterministic algorithm on a randomised input.

\begin{lemma}[Corollary of Yao's minimax principle \cite{yao1977probabilistic}]
\label{lem:yao}
For any distribution $D$ on the inputs, for any randomised algorithm $M$, the expected query complexity of $M$ is at least the average query complexity of the best deterministic algorithm for input distribution $D$. 
\end{lemma}
We will apply Yao's principle for $D$ the uniform distribution and the query complexity measuring the number of queries answered \textbf{NO}. We combine this with the following lemma.

\begin{restatable}{lemma}{concentrationcoord}
\label{lem:concentration_coord}
Let $n,k$ and $\Delta$ be integers.
    For any deterministic algorithm $R$ using the coordinate oracle and $f: [n] \rightarrow [\Delta]^k$ sampled u.a.r., the probability that $R$ reconstructs $f$ in at most $\frac1{10} \Delta n k$ queries answered \textbf{NO} is at most $e^{-\frac{1}{50} n k}$.
\end{restatable}

We first deduce our main theorem on function reconstruction from the two lemmas above.

\begin{proof}[Proof of \cref{thm:lb_coord}]
Let $M$ be a deterministic algorithm that  reconstructs balanced functions using the coordinate oracle. We first extend $M$ to an algorithm $\widetilde{M}$ that reconstructs all functions (among all functions) while the number of \textbf{NO} answers remains the same if the input is balanced. The algorithm $\widetilde{M}$ first performs the same queries as $M$ does, until it either has no balanced candidates or a single balanced candidate $f$ compatible with the answers so far. In the former case, it reconstructs the function by brute-force. In the second case, it performs $\qu_1^c(a,f(a)_i,i)$ for all $a\in A$ and $i\in [k]$ to verify that indeed the input is $f$. If the input is indeed $f$, we have now distinguished $f$ among all functions (rather than all balanced functions) without additional \textbf{NO} answers. If any of the queries answers \textbf{NO}, we again have no balanced candidates left and may perform the brute-force approach again.

We will show that, when restricted to balanced functions, $\widetilde{M}$ has an average query complexity (in terms of the number of \textbf{NO} answers) greater than $\frac{1}{11}\Delta n k$. Since $M$ has the same number of \textbf{NO} answers as $\widetilde{M}$ on balanced inputs, it has the same average query complexity as $\widetilde{M}$. 
Using Yao's principle (\cref{lem:yao}), it then follows that any randomised algorithm that reconstructs balanced functions has at least $\frac{1}{11}nk$ queries answered \textbf{NO} in expectation.  
 
By \cref{lem:concentration_coord}, there are at most $|B|^n e^{-\frac1{50}n k}$ functions $f:A\to B$ for which $\widetilde{M}$ reconstructs $f$ in less than $\frac1{10}\Delta n k$ queries. On the other hand the number of balanced function from $A$ to $B$ is the following multinomial coefficient ${n \choose c,\ldots,c} = \frac{n!}{(c!)^n}$. In particular, there are at least ${n \choose c,\ldots,c} - (n/c)^n e^{-\frac1{50}nk}$ balanced function for which $\widetilde{M}$ requires at least $\frac1{10}\Delta nk$ queries. This means that the average query complexity of $\widetilde{M}$ is at least
\[
\frac{{n \choose c,\ldots,c} -(n/c)^n e^{-\frac1{50}nk}}{{n \choose c,\ldots,c}}\frac1{10}\Delta n k = \frac{n! - (c!)^n (n/c)^n e^{-\frac1{50}nk}}{n!}\frac1{10}\Delta n k \geq \frac1{11}\Delta n k
\]
since,
\[
\displaystyle
(c!)^n (n/c)^ne^{-\frac1{50}nk} = \left(\frac{n}e\right)^n \left( \frac{ec!}c e^{-\frac1{50}k}\right)^n \leq n^n e^{-\frac{51}{50}n}\leq  \frac1{100}n!\]
using for the first inequality that $k \geq 50(c \ln c + 2)$ and for the second that $n\geq 2^{51}$.
\end{proof}

\begin{proof}[Proof of Lemma \ref{lem:concentration_coord}]
Let $R$ be a deterministic algorithm that uses the coordinate oracle to reconstruct functions.
Let $F_t$ denote the set of possible functions $f:A\to B$ that are consistent with the first $t$ queries done by $R$. (This depends on the input function $g:A\to B$, but we leave this implicit.) For $a\in A$ and $i\in [k]$, let
\[
J_{a,i}^t = \{j\in [\Delta] \mid f(a)_i=j \text{ for some }f\in F_t\}.
\]
Note that all values $j_1,j_2\in J_{a,i}^t$ are equally likely in the sense that there is an equal number of $f\in F_t$ with $f(a)_i=j_1$ as with $f(a)_i=j_2$.
The algorithm $R$ will perform the same $t$ queries for all $f\in F_t$. In particular, if $g:A\to B$ was chosen uniformly at random, then after the first $t$ queries all $f\in F_t$ are equally likely (as input function) and in particular $g(a)_i$ is uniformly distributed over $J_{a,i}^t$, independently of the sets $J_{a',i'}^t$ for $(a',i')\neq (a,i)$.
This is the part for which we crucially depend on the fact that we allow all functions $f:A\to B$ and not just bijections (where there may be dependencies between the probability distributions of $g(a)$ and $g(a')$ for distinct $a,a'\in A$). 

We say that the $t^\text{th}$ query of the algorithm is \emph{special} if 
    \begin{itemize}
        \item it is a Type 1 query $\qu_1^c(a,b,i)$ and  $|J_{a,i}^t| \geq \Delta/2$, or
        \item it is a Type 2 query $\qu_2^c(a,a',i)$ and either $|J_{a,i}^t|$ or $|J_{a',i}^t|$ is at least $\Delta/2$.
    \end{itemize} 
    Let $T$ denote the number of  \textbf{NO} answers to \emph{special} queries that $R$ does to the coordinate oracle until it has reconstructed the input function.
    We let $Y_i=1$ if the answer of the $i^\text{th}$ special query is \textbf{YES} and $0$ otherwise. So $\sum_{i=1}^T Y_i$ denotes the number of special queries with answer \textbf{YES}.
    
    At the start of the algorithm $J_{a,i}^0 = [\Delta]$ for all $a\in A$ and $i \in [k]$. Thus, to reconstruct the function, the pair $(a,i)$ is either (1) involved in a special query with answer is \textbf{YES} or (2) involved in $\Delta/2$ special queries for which the answer is \textbf{NO}. 
Since any query involves at most two elements of $A$, we deduce that 
    \[
    |A|k/2 = nk/2 \leq 
    \left(T - \sum_{i=1}^{T} Y_i \right) \frac{2}{\Delta} + \sum_{i=1}^{T} Y_i.
    \]
    We aim to prove that if $g:A\to B$ is sampled uniformly at random, then with high probability $T=T(g) \geq \frac{1}{10}\Delta nk$. In order to do so, we consider a simplified process and a random variable $\tau$ which is stochastically dominated by $T$ (i.e. for any $x \in \mathbb{R}^+$, $\Pr(T \leq x) \leq \Pr(\tau \leq x)$). Let us consider an infinite sequence of i.i.d. random variables $X_1,X_2,X_3,\dots  \sim \text{Bernouilli}(2/\Delta)$. Note that
    $$
    H(t)=\left( t - \sum_{i=1}^t X_i \right) \frac{2}{\Delta} + \sum_{i=1}^t X_i = \left(1- \frac2{\Delta} \right) \sum_{i=1}^t X_i + \frac{2t}{\Delta}
    $$
    is increasing in $t$. Let $\tau$ be the first integer $t$ for which $H(t)\geq \frac1{2}n \log_\Delta n$.
    
    If $g$ is sampled uniformly at random then the $j$th special query 
    (say involving $a\in A$ and $i\in [k]$ with $|J^t_{a,i}|\geq\Delta/2$)
    has answer \textbf{YES} with probability
    \[
    \Pr(Y_i = 1) \leq \frac{2}\Delta = \Pr(X_i = 1).
    \]
    This is because all values of $J^t_{a,i}$ are equally likely for $g(a)_i$ (and independent of the value of $g(a')_i$ or $b_i$ for $b\in B$ and $a'\in A$). This inequality holds independently of the values of $(Y_1,\dots,Y_{i-1})$. This implies that, for any $t \in \mathbb{N}^+$ and any $x \in \mathbb{R}^+$, 
    $$
    \Pr\left(\sum_{i=1}^t X_i \leq x\right) \leq \Pr\left(\sum_{i=1}^t Y_i \leq x\right).
    $$
    Therefore,
    $$
    \Pr\left(\left(1- \frac2{\Delta}\right) \sum_{i=1}^t X_i + \frac{2t}{\Delta} \leq x\right) \leq \Pr\left(\left(1- \frac2{\Delta}\right) \sum_{i=1}^t Y_i + \frac{2t}{\Delta} \leq x\right).
    $$
    From this we can conclude that $\Pr(T \leq x) \leq \Pr(\tau \leq x)$, thus $T$  stochastically dominates $\tau$.

    If $\tau \leq \frac1{10}\Delta nk$, then using the definition of $\tau$ we find that
    \begin{align*}
        \left( \frac1{10}\Delta nk - \sum_{i=1}^\tau  X_i \right) \frac{2}{\Delta} + \sum_{i=1}^\tau X_i &\geq \frac12 nk\\
        \intertext{which implies}
          \sum_{i=1}^\tau X_i\geq \left(1 - \frac{2}{\Delta}\right) \sum_{i=1}^\tau X_i & \geq \frac3{10}n k. 
    \end{align*}
Let $x =  \frac1{10}\Delta nk$. We compute $\mathbb{E}\left[\sum_{i=1}^x X_i\right] = \frac{2}{\Delta}x = \frac15 nk$. Using Chernoff's inequality (see e.g. \cite{mitzenmacher2017probability}) we find
\[
 \Pr(\tau \leq x) \leq \Pr\left(\sum_{i=0}^{x} X_i \geq \left(1+\frac12 \right) \frac15 nk\right) \leq \exp\left(-\left(\frac12\right)^2\frac15 n k /\left(2+\frac12\right)\right).
\]  
Since $\frac12 \frac12\frac15 \frac2{5}=\frac1{50}$, this proves $\Pr(T \leq x) \leq \Pr(\tau \leq x)\leq e^{-\frac1{50}n k}$. In particular, the probability that at most $\frac1{10}\Delta k$ queries are used is at most $e^{-\frac1{50}n k}$, as desired.
\end{proof}

\subsection{Reconstructing functions from the word oracle}
\label{subsec:word_oracle}

Let once again $A =[n]$ and $B = [\Delta]^k$. We next turn our attention to reconstructing functions $f:A\to B$ from a more complicated oracle that we use as a stepping stone to get to distance queries in trees.
For $b\in B$, we write $b_{[i,j]}=(b_i,b_{i+1},\dots,b_j)$. It will also be convenient to define $b_{\emptyset}$ as the empty string. 
The \textbf{word oracle}  can answer the following two types of questions.
\begin{itemize}[leftmargin=*]
    \item \textbf{Type 1}. $\qu_1^w(a,b)$ for $a \in A$ and $b \in B$, answers the largest $i\in [0,k]$ with $f(a)_{[1,i]}= b_{[1,i]}$.
    \item \textbf{Type 2}. $\qu_2^w(a,a')$ for $a,a' \in A$, answers the largest $i\in [0,k]$ with $f(a)_{[1,i]}= f(a')_{[1,i]}$.
\end{itemize}
By studying the number of queries for the word oracle and the number of \textbf{NO} answers for the component oracle, we can link the two reconstruction problems as follows.

\begin{restatable}{lemma}{coordword}
    \label{lem:coord_word}
    For all positive integers $\Delta,k$ and $n$, for any algorithm $M$ using the word oracle that reconstructs functions $f:A \rightarrow B$ in at most  $q(f)$ queries in expectation, there exists an algorithm $M'$ using the coordinate oracle that reconstructs functions $f:[n] \to [\Delta]^k$ such that at most $q(f)$ queries are answered \textbf{NO} in expectation.
\end{restatable}

\begin{proof}
    Given an algorithm $M$ using the word oracle, we build a new algorithm $M'$ using the coordinate oracle. We do so query-by-query. If $M$ asks $\qu_1^w(a,b)$, then $M'$ performs a sequence of queries $\qu_1^c(a,b_1,1),\qu_1^c(a,b_2,2),\ldots,\qu_1^c(a,b_{i+1},i+1)$, where $i\in [0,k-1]$ is the largest for which $f(a)_{[1,i]} = b_{[1,i]}$. Note that the sequence indeed simulates a query of the word oracle yet  the coordinate oracle answers \textbf{NO} at most once (on the $(i+1)$th query). 
    
   Queries of Type $2$ can be converted analogously.  This way, for every input $f$, the natural `coupling' of the randomness in $M$ and $M'$ ensures that the number of \textbf{NO} answers to $M'$ is stochastically dominated by the number of queries to $M$. In particular, the expected number of \textbf{NO} answers given by $M'$ is upper bounded by $q(f)$, the expected number of queries to~$M$.
\end{proof}
\cref{lem:coord_word,thm:lb_coord} now give the following result.

\begin{corollary}
    \label{lem:lb_word}
    Let $\Delta \geq 3$, $c\leq \Delta-1$ and $k \geq 50(c\ln c +2)$ be positive integers. Let $n = c\Delta^k$. Any algorithm reconstructing $f:[n] \to [\Delta]^k$ using the word oracle, in the special case where $f$ is known to be balanced, needs at least $\frac1{11}\Delta n k$ queries in expectation.
\end{corollary}

\subsection{Reducing tree reconstruction to function reconstruction}
\label{subsec:reduction_to_function}
\begin{figure}
    \centering
    \input{fig/treelowerbound}
    \caption{Example of the tree $T_{c,\Delta,k}$ constructed in the proof of \cref{thm:lowerboundnlogn} for $\Delta = 4$, $c=1$ and $k=2$ with the labelling $\ell$ of the internal nodes.}
    \label{fig:treelowerbound}
\end{figure}

In order to prove \cref{thm:lowerboundnlogn} we consider a specific tree $T_{c,\Delta,k}$ (with $c \leq \Delta$): the tree of depth $k+1$ where each node at depth at most $k-1$ has exactly $\Delta$ children and each node at depth $k$ has exactly $c$ children (see \cref{fig:treelowerbound}). \cref{thm:lowerboundnlogn} is an almost direct consequence of the following lemma. 

\tleaflabel*

\begin{proof}
We consider $T = T_{c,\Delta,k}$ and let $L$ be the set of leaves of $T$ and let $P$ be the set parents of the leaves. The tree $T$ has $n=\sum_{i=0}^k\Delta^i + c\Delta^k$ nodes and $\n =c\Delta^k$ leaves. Let $p:L\to P$ be the bijection that sends each leaf to its direct parent. 
We label internal nodes as follows. The root is labelled $\emptyset$ (the empty string) and if a node $v$ has label 
$\ell$ and has $\Delta$ children, then we order the children $1,\dots,\Delta$ and we label the child $i$ with label obtained from concatenation $\ell+(i)$. We put such labels on all internal nodes.

Let $I$ denote the set of internal nodes and let $\ell(v)$ denote the label of $v\in I$.
Let $f:L\to [\Delta]^k$ be the bijection that sends a leaf $u\in L$ to the label $\ell(p(u))\in [\Delta]^k$ of its direct parent. 

We consider the trees which have a fixed labelling (as described above) for node in $I$, and every possible permutation of the labelling of the leafs. 
All possible bijections $f:L\to [\Delta]^k$ appear among the trees that we are considering. To reconstruct the tree, we in particular recover the corresponding bijection $f$. Distance queries between internal vertices always give the same response and can be ignored. We show the other queries are Type 1 and Type 2 queries in disguise.
\begin{itemize}
    \item  For $a \in L$ and $b\in I$ the distance between $a$ and $b$ is given as follows. Let $z\in I$ be the nearest common ancestor of $a$ and $b$ and say $z$ has depth $i$ and $b$ has depth $j$. The distance between $a$ and $b$ is $1+(k-i)+(j-i)$. The values of $k$ and $j$ do not depend on $f$ but the value of $i$ is exactly given by $\max\{s:f(a)_{[0,s]}= b_{[0,s]}\}$, the answer to the corresponding type 1 query of $(a,b)$ to the word oracle. To be precise, since $b$ may have a length shorter than $i$, the query $\qu_1^w(a,b')$ where $b'_s=b_s$ for all $s\in [1,|b|]$ and $b_s'=0$ otherwise,  gives the desired information.
    \item For $a,a' \in L$, the distance between $a$ and $a'$ is given by $2(1+(k-i))$ for $i$ the answer of a Type 2 $\qu_2^w(a,a')$ to the word oracle. 
\end{itemize}

This shows that we reduce an algorithm to reconstruct the labelling of the leaves from $q$ distance queries to an algorithm that reconstructs functions $f: L \to [\Delta]^k$ from $q$ queries to the word oracle. By \cref{thm:lb_coord}, since $k \geq 50(c \ln c +2)$, we need at least $\frac1{11}\Delta \n k$ queries. 
\end{proof}

We are now ready to deduce the main result of this section.

\lowerboundnlogn*

\begin{proof}
    Let $\Delta,n\geq 2$ be integers. We write $n=2c\Delta^k$ for $c\in [1,\Delta)$ and $k$ an integer. (When $n/2\geq 1$, there is a unique pair $(c,k)\in [1,\Delta)\times \mathbb{Z}_{\geq 0}$ with $n/2=c\Delta^k$.) 
    
    Suppose that $k\geq 50(c\ln c+3)$. In particular, $\Delta \geq 2$ implies $k \geq \lfloor \log_\Delta n - 1\rfloor$.
    The tree $T=T_{\lfloor c \rfloor,\Delta,k}$ considered in \cref{lem:t_leaf_label} has maximum degree $\Delta+1$,  $\n =\lfloor c \rfloor \Delta^{k}$ leaves and $n' = \sum_{i=0}^{k} \Delta^i +\lfloor c \rfloor \Delta^{k}$ vertices, where
    \[
    n/4\leq N\leq  n' \leq 2 {c} \Delta^k = n.
    \]
    For $\Delta\geq 2$ and $c \geq 1$, if $n \geq 2\Delta^{50(c \ln c+3)}$ then $\n \geq \frac{n}4 \geq \Delta^{50(c \ln c +2)}$. So we may apply \cref{lem:t_leaf_label} and find that at least 
    \[
        \frac1{11} \Delta \n \lfloor \log_{\Delta} \n -1 \rfloor \geq \frac1{44} \Delta n ( \log_{\Delta} n-4)
    \]
    queries are required.  
    As $\log_\Delta n\geq 150$, we find that this is at least  $\frac1{50}\Delta n \log_{\Delta} n$.
\end{proof}

\begin{corollary}
Any randomised algorithm requires at least $\frac1{50} \Delta n\log_\Delta n$ distance queries to reconstruct $n$-vertex trees of maximum degree $\Delta+1\geq 3$ if $n\geq 2\Delta^{50 (\Delta \ln \Delta+3)}$.
\end{corollary}

\subsection{Randomised lower bounds for related models} 
\label{subsec:phylogentic}
We next show that our result implies various other new randomised lower bounds. Although we state these results with a weaker assumption on $n$ for readability reasons, we remark that our more precise set-up (allowing $\Delta$ to be a small polynomial in $n$ for specific values of $n$) also applies here. 

\paragraph{Betweenness queries}
A \emph{betweenness query} answers for three vertices $(u,v,w)$ whether $v$ lies on a shortest path between $u$ and $w$. 
Using three distance queries to $(u,w)$, $(u,v)$ and $(v,w)$, you can determine whether $v$ lies on a shortest path between $u$ and $w$, so the betweenness oracle is weaker (up to multiplicative constants).
It has been shown in \cite{abrahamsenBodwinRotenbergStockel16} that randomised algorithms can obtain a similar query complexity for betweenness queries as was obtained for distance queries by \cite{KannanMZ14}. Moreover, a randomised algorithm for $4$-chordal graphs has been given that uses a quasi-linear number of queries to a betweenness oracle \cite{RONG20221}. A deterministic algorithm using $\tilde{O}(\Delta n^{3/2})$ betweenness queries has been given for trees, as well as a $\Omega(\Delta n)$ lower bound \cite{TreesSeparatorQueriesJagadishSen}. 
Our randomised lower bound from Theorem \ref{thm:lowerboundnlogn} immediately extends to this setting.
\begin{corollary}
Any randomised algorithm requires at least $\frac1{150} \Delta n\log_\Delta n$ betweenness queries to reconstruct $n$-vertex trees of maximum degree $\Delta+1\geq 3$ if $n\geq 2\Delta^{50 (\Delta \ln \Delta+3)}$.
\end{corollary}

\paragraph{Path and comparison queries}
Given two nodes $i,j$ in a directed tree, a \textit{path query} answers whether there exists a directed path from $i$ to $j$. 
Improving on work from \cite{Wang19}, it was shown in \cite{AfsharESA20} that any algorithm needs $\Omega(n\log n+ n\Delta)$ to reconstruct a directed tree on $n$ nodes of maximum degree $\Delta$.
When we consider a directed rooted tree in which all edges are directed from parent to child, then path queries are the same as \emph{ancestor queries}: given $u,v$ in a rooted tree, is $u$ an ancestor of $v$? 
We apply this to the tree $T_{c,\Delta,h}$ from \cref{lem:t_leaf_label} for which the labels of all internal vertices are fixed but the labels of the leaves are unknown. Path queries $(u,v)$ only give new information if $v$ is a leaf and $u$ is an internal vertex. But this is weaker than distance queries, since we can obtain the same information by asking the distance between $u$ and $v$. This means that we can redo the calculation from the proof of \cref{thm:lowerboundnlogn} (applying \cref{lem:t_leaf_label})  to lift the lower bound to path queries.
\begin{corollary}
Any randomised algorithm requires at least $\frac1{50} \Delta n\log_\Delta n$ path queries to reconstruct $n$-vertex directed trees of maximum degree $\Delta+1\geq 3$ if $n\geq 2\Delta^{50(\Delta \ln \Delta+3)}$. 
\end{corollary}
A randomised algorithm using $O(n\log n)$ path queries on bounded-degree $n$-vertex trees has been given in \cite{AfsharESA20} but their dependency on $\Delta$ does not seem to match our lower bound. 
We remark that besides query complexity, works on path queries such as \cite{AfsharESA20,Afshar22LATIN,Wang19} also studied the round complexity (i.e. the number of round needs when queries are performed in parallel).  

The same ideas applies to lift our lower bound to one for reconstructing tree posets  $(T,>)$ from \emph{comparison queries}, which answer for given vertices $u,v$ of the tree whether $u<v,v<u$ or $u||v$.
\begin{corollary}
Any randomised algorithm requires at least $\frac1{50} \Delta n\log_\Delta n$ comparison queries to reconstruct $n$-vertex tree posets of maximum degree $\Delta+1\geq 3$ if $n\geq 2\Delta^{50 (\Delta \ln \Delta+3)}$. 
\end{corollary}
This improves on the lower bound of $\Omega(\Delta n + n \log n)$ from \cite{Roychoudhury23} and matches (up to a $(C\log\Delta)$ factor) the query complexity of their randomised algorithm. 

\paragraph{Membership queries for reconstructing partitions}
The $(n,k)$-partition problem was introduced by King, Zhang and Zhou~\cite{KingZhangZhou}.  Given $n$ elements which are partitioned into $k$ equal-sized classes, the partition needs to be determined via queries of the form `Are elements $a$ and $b$ in the same class?'.  They used the adversary method to prove $\Omega(nk)$ queries are needed by any deterministic algorithm. 
Liu and Mukherjee~\cite{Liu22} studied this problem phrased as learning the components of a graph via membership queries (which answer whether given vertices lie in the same component or not) and provide an exact deterministic lower bound of $(k-1)n-\binom{k}2$ for deterministic algorithms. 
It indeed seems natural that randomised algorithms need $\alpha kn$ queries for some constant $\alpha$ in this setting. Nonetheless the best lower bound for randomised algorithm seems to be the information-theoretic lower bound of $\Omega(n\log k)$. Our next result remedies this gap in the literature.
\begin{corollary}
Let $\varepsilon>0$.
Any randomised algorithm requires at least $\frac1{11} nk$ membership queries to solve the $(n,k)$-partition problem if $n\geq k^{1+\varepsilon}$ is sufficiently large.
\end{corollary}
\begin{proof}
Since we plan to apply \cref{lem:concentration_coord}, we will write $\Delta=k$.

First note that we can see the $(n,\Delta)$-partition problem using \emph{membership} queries as reconstructing a balanced function using only Type 2 queries to the coordinate oracle. Formally, if $f : [n] \to [\Delta]$ is the function which associates an element $a\in [n]$ to the index $i\in [\Delta]$ of the part that contains $a$ (out of the $\Delta$ parts in the partition), then a membership query between $a,a'\in [n]$ is exactly equivalent to the coordinate query $\qu_2^{c}(a,a')$ applied the function $f$. Once we reconstructed the partition, we can retrieve the index of  each parts using $\Delta^2 = o(n\Delta)$ queries of Type 1 to the coordinate oracle. Therefore it suffices to show that at least $\frac{1}{11}\Delta n $ queries are needed in expectation to reconstruct a balanced function using the coordinate oracle.

Applying \cref{lem:concentration_coord} with $k=1$, we find that when $f$ is sampled uniformly at random (among all functions $g:[n]\to [\Delta]$), the probability that a given randomised algorithm uses less than $\frac1{10}n\Delta$ queries is at most $e^{-\frac{1}{50} n}$. In particular, the number of balanced functions reconstructed in less than $\frac1{10}n\Delta$ queries is upper bounded by $\Delta^n e^{-\frac{1}{50} n}$. We compare this number to the total number ${n \choose n/\Delta,\ldots,n/\Delta} = \frac{n!}{(n/\Delta)!^\Delta}$ of balanced functions:
\begin{align*}
    \frac{\Delta^n e^{-\frac{1}{50} n}}{\frac{n!}{(n/\Delta)!^\Delta}} &\leq 
\Delta^n    e^{-\frac1{50}n} \frac{(2\pi n/\Delta)^{\Delta/2}(n/(e\Delta)^{n}}{\sqrt{2\pi n}(n/e)^n}
e^{\Delta/(12n/\Delta)}\\
&=\frac1{\sqrt{2\pi n}}(2\pi n/\Delta)^{\Delta/2} \exp(\Delta^2/(12n)-n/50).
\end{align*}
Here we used that for all $n\geq 1$, 
\[
\sqrt{2\pi n} (n/e)^n <n!< \sqrt{2\pi n} (n/e)^n e^{1/(12n)}.
\]
Since $\Delta\leq n^{1/(1+\epsilon)}$ for some $\epsilon>0$, the fraction tends to $0$, so in particular becomes smaller than $\frac1{100}$ when $n$ is sufficiently large (depending on $\epsilon$).
This implies that the expected number of queries to reconstruct a $\Delta$-balanced function  is at least $\frac{99}{100} \frac1{10} \Delta n \geq \frac{1}{11} n \Delta$. By the discussion at the start, we find the same lower bound for the $(n,\Delta)$-partition problem.
\end{proof}
The same lower bound holds when queries of the form `Is element $a$ in class $i$?' are also allowed. We expect that our methods can be adapted to handle parts of different sizes and that our constant $\frac1{11}$ can be easily improved.

Randomised algorithms have been studied in a similar setting by
Lutz, De Panafieu, Scott and Stein \cite{scottNIPS} under the name \emph{active clustering}. They provide the optimal average query complexity when the partition is chosen uniformly at random among all partitions. They also study the setting in which a partition of $n$ items into $k$ parts is chosen uniformly at random and allow queries of the form `Are items $i$ and $j$ in the same part?'. However there is a key difference: the set of answers the algorithm receives, needs to distinguish the partition from any other partition (including those with a larger number of parts). In this setting, the following algorithm is optimal. Order the items $1,\dots,n$. For $i=1,\dots,n$, query item $i$ to items $j=1,\dots,i$ in turn if the answer to the query `Are items $i$ and $j$ in the same part?' is not yet known. It follows from \cite[Lemma 9]{scottNIPS} that this algorithm is has the lowest possible expected number of queries. The expected number of queries used is at most 
\[
(1+2+3+\dots +k)\frac1k n=\frac{k+1}2n.
\]
In particular, for $i=\Omega(\log n)$ queries, with high probability there are items in $k$ different parts among the first $i-1$ items. Since the number of parts is not `known', the algorithm will use $k$ queries for item $i$ if is in the `last part' and so the complexity is $(\frac{k+1}2+o(1))n$ as $n\to \infty$. The same algorithm would use $(\tfrac{k+1}2-\tfrac1k+o(1)) n$ queries when the number of parts may be assumed to be at most $k$, in which case queries `to the last part' are never needed. So the assumption on whether the number of parts is known, changes the query complexity. Together with the parts `being known to be exactly balanced', this introduces additional dependencies in the $(n,k)$-partition problem that our analysis had to deal with.

\paragraph{Phylogenetic reconstruction}
This setting comes from biology. Reconstructing a phylogenetic tree has been modelled via what we call \emph{leaf-distance} queries (similarity of DNA) between leaves of the input tree \cite{hein1989optimal,waterman1977additive,KingZhangZhou}. Although very similar, the query complexities of the phylogenetic model and the distance query model on trees are not directly related. In the phylogenetic model, the set of leaves is already known and the leaf-distance queries are only possible between leaves. Moreover, we consider a phylogenetic tree to be reconstructed once we know all the pairwise distances between the leaves. For example, if the input tree is a path on $n$ vertices, then in the phylogenetic setting we receive only two leaves and are finished once we query their distance, whereas in the distance reconstruction setting it takes $\Omega(n)$ queries to determine the exact edge set.

Improving on various previous works~\cite{brodal01,Kannan96,KaoLingasOstlin99}, King, Zhang, and Zhou~\cite[Theorem 3.2]{KingZhangZhou} showed that any deterministic algorithm reconstructing phylogenetic trees of maximum degree $\Delta$ with $\n$ leaves needs at least $\Omega(\Delta \n \log_\Delta \n)$ leaf-distance queries. Deterministic algorithms achieving this complexity are also known~\cite{hein1989optimal}.

For randomised algorithms, the previous best lower bound was the information-theoretic lower bound of $\Omega(\n\log \n/\log\log \n)$. 
We provide the following randomised lower bound from~\cref{lem:t_leaf_label} which is tight up to a multiplicative constant.
\begin{restatable}{theorem}{phylogenetic}
\label{thm:phylogenetic}
Let $\Delta\geq 2,~c\leq \Delta-1$ and $k\geq 50(c\ln c+2)$ be positive integers. Let $\n=c\Delta^k$.
Any randomised algorithm reconstructing phylogenetic trees of maximum degree $\Delta+1$ with $\n$ leaves needs at least $\frac{1}{20}\Delta \n \log_\Delta \n$ leaf-distance queries in expectation.    
\end{restatable}
\begin{proof}   
Let $T=T_{c,\Delta,k}$ be the tree considered in \cref{lem:t_leaf_label} with $\n$ leaves. 

Suppose towards a contradiction that we could obtain the pairwise distances between the leaves of this tree in $\frac1{20}\Delta \n\log_\Delta \n$ leaf-distance queries in expectation.
We show that, from this, we can recover the labels of the leaves of $T$ using only $\Delta^2 \n \leq \frac{1}{30} \Delta \n \log_\Delta \n$ additional distance queries, contradicting \cref{lem:t_leaf_label} since $\frac1{20}+\frac1{50}\leq \frac1{11}$ and $\log_\Delta \n \geq k \geq 50$. 

We proceed by induction on $k$, the depth of the tree $T$.
When $k=0$, $T_{c,\Delta,0}$ is a star with $c$ leaves. There is nothing to prove, as the parent of each leaves in known to be the root. Suppose $k \geq 1$ and that the claim has already been shown for smaller values of $k$.
We define an equivalence relation on the set of leaves: for $u_1,u_2 \in L$, $u_1 \sim u_2$ if and only if $d(u_1,u_2) < 2k$. This is an equivalence relation with $\Delta$ equivalence classes, as $d(u_1,u_2) < 2k$ if and only if $u_1$ and $u_2$ have a child of the root as common ancestor. 

Let $u_1,u_2,\ldots,u_\Delta$ be arbitrary representatives of each of the $\Delta$ equivalence classes. (Note that we can select these since we already know the distances between the leaves.)
Let $r$ denote the root of $T$. 
We ask $\qu(u_i,N(r))$ for all $i \in [\Delta]$. From this we can deduce the common ancestor among the children of the root for each of the classes. It is the unique neighbour of $r$ lying on a shortest path from $u_i$ to $r$. Let  $V_i$ denote the set of all the leaves that have the $i^\text{th}$ neighbour of $r$ as common ancestor. We also define $T_i$ to be the subtree rooted in the $i^\text{th}$ neighbour of $r$. We remark that it is now sufficient to solve $\Delta$ subproblems of reconstructing $T_i$ knowing each $V_i$ leaf matrix. By the induction hypothesis, each subproblem is solvable in $|V(T_i)| \Delta^2 = \frac{\n-1}{\Delta}\Delta^2 = (\n-1)\Delta$ queries. Therefore, in total this algorithm uses $(\n-1)\Delta^2 + \Delta^2 = \Delta^2 \n$ distance queries.
\end{proof}

\section{Open problems}
\label{sec:concl}
We presented new algorithms for reconstructing classes of graphs with bounded maximum degree from a quasi-linear number of distance queries and gave a new randomised lower bound of $\const \Delta n \log_{\Delta} n$ for $n$-vertex trees of maximum degree $\Delta$. This lower bound is now also the best lower bound for the class of bounded degree graphs, while  the best-known randomised algorithm uses $\widetilde{O}_\Delta(n^{3/2})$ queries \cite{mathieu2013graph}. 
The main open question is to close the gap between our lower bound and this upper bound. In particular, it would be interesting to see if the quasi-linear query complexity achieved for various classes of graphs can be extended to all graphs of bounded maximum degree. 
\begin{problem}
\label{problem1}
Does there exist a randomised algorithm that reconstructs an $n$-vertex graph of maximum degree $\Delta$ using $\Tilde{O}_{\Delta}(n)$ distance queries in expectation?
\end{problem}

From a more practical viewpoint, studying the query complexity of reconstructing scale-free networks is of high importance as it is the class of graphs that best describes real-world networks like the Internet. Graphs in this class have vertices of large degree, therefore recent theoretical works (including this one), do not directly apply. In particular, not even an $o(n^2)$ algorithm is known in this setting.

\begin{problem}
How many distance queries are needed for reconstructing scale-free networks?
\end{problem}

We showed in this paper that \textit{deterministic} algorithms with good query complexity exist for specific classes of graphs. One of the classes of graphs that does not fit in the scope of this paper but is known to have an efficient randomised algorithm is the class of bounded degree outerplanar graphs \cite{KannanMZ14}. Note that outerplanar graphs also have a nice separator structure. For example, there is always a $\frac12$-balanced separator which induces a path (see e.g. \cite[Proposition 6]{cyril}).

\begin{problem}
    Does there exist a deterministic algorithm that reconstructs an $n$-vertex outerplanar graphs of maximum degree $\Delta$ in $\Tilde{O}_{\Delta}(n)$ distance queries?
\end{problem}

We provided various randomised lower bounds in Section \ref{sec:lb} that are tight up to a multiplicative constant.
We expect that determining the exact constant for the optimal expected number of queries for reconstructing $n$-vertex trees of maximum degree $\Delta$ for the entire range of $(n,\Delta)$ may be complicated. 
We believe that for deterministic algorithms, when $n$ is large compared to $\Delta$ (and takes particularly nice forms, e.g. $1+\Delta+\dots+\Delta^k$ for some integer $k\geq 1$), determining the correct constant could be achievable (yet would require new ideas). 
In particular, our simple algorithm may be even optimal up to an additive constant. We set the following challenge towards this.
\begin{problem}
Is there a constant $c>\frac12$ such that for all $\Delta\geq 3$, there are infinitely many values of $n$ for which any deterministic algorithm that reconstructs $n$-vertex trees of maximum degree $\Delta$ needs at least $c\Delta n\log_\Delta n$ queries?  
\end{problem}
We showed that the randomised and deterministic query complexity have the same dependence on $n$ and $\Delta$ and in that sense randomness does not help much. If the answer to the question above is positive, it shows that  randomness does at least make some difference.

A lower bound for trees immediately implies a lower bound for any class containing  trees (such as $k$-chordal graphs), but a better lower bound could hold for bounded degree graphs. 
In our algorithm for $k$-chordal graphs, we did not try to optimise the dependency on the maximum degree $\Delta$ and $k$ and the dependency on $k$ can probably be improved upon. 
If one believes the answer to \cref{problem1} to be positive, then the dependency on $k$ should be at most poly-logarithmic.
It would be interesting to see if some dependency on $k$ is needed. To ensure the lower bound needs to exploit cycles due to our $\Delta n \log_\Delta n$ algorithm for trees, we pose the following problem.
\begin{problem}
    Is it true that for some fixed values of $k$ and $\Delta$ and for all $n$ sufficiently large, any algorithm reconstructing $k$-chordal graphs on $n$ vertices of maximum degree $\Delta$ requires at least $10^6 \Delta n\log_\Delta n$ queries?
\end{problem}
Finally, we believe the problems of reconstructing functions from the coordinate or from the word oracle are of independent interest. There are various variations that can be considered  to which our methods may extend. For example, it is also natural to consider (bijective) functions $f:A\to B$ where both $A$ and  $B$ are the same product of sets of different sizes (e.g. $A=B=[a_1]\times [a_2]\times \dots \times [a_k]$).

\paragraph{Acknowledgements} We would like to thank Claire Mathieu for helpful suggestions and Guillaume Chapuy for helpful discussions regarding randomised lower bounds for the $(n,k)$-partition problem which inspired our randomised lower bound proofs. We are also grateful to Jatin Yadav for pointing us to an issue we did not address in an earlier version and we thank an anonymous referee for pointing out that our upper bound in Theorem \ref{thm:tree_rec} is different when $\Delta=3$.

\newcommand{\etalchar}[1]{$^{#1}$}

\end{document}

%% file: preamble.tex
\usepackage[utf8]{inputenc}
\usepackage[english]{babel}
\usepackage[T1]{fontenc}
\usepackage[letterpaper, margin=3cm]{geometry}
\usepackage{needspace}
\usepackage{bbm}

\usepackage{libertine} 
\usepackage{inconsolata} 

\usepackage[official]{eurosym}
\usepackage{setspace}

\usepackage{stackengine}
\newcommand\fit[3][.3ex]{\stackengine{#1}{#3}{#2}{O}{c}{F}{T}{S}}
\newcommand\umlaut[1]{\fit[.05ex]{\scriptsize..}{#1}}

\usepackage{amsmath, amsthm, amssymb}
\usepackage{graphicx}
\usepackage{enumerate}
\usepackage{enumitem}
\usepackage{authblk}
\usepackage{thmtools}
\usepackage{thm-restate}
\usepackage[colorlinks=true, citecolor=red]{hyperref}
\usepackage[capitalize]{cleveref}
\usepackage{float}
\usepackage{amsfonts,mathtools,enumerate}
\usepackage[ruled,vlined]{algorithm2e}
\usepackage{url} 

\usepackage{tikz}
\usetikzlibrary{arrows,decorations,backgrounds,snakes,calc} 

\renewenvironment{abstract}
{\small\vspace{-1em}
\begin{center}
\bfseries\abstractname\vspace{-.5em}\vspace{0pt}
\end{center}
\list{}{
\setlength{\leftmargin}{0.6in}%
\setlength{\rightmargin}{\leftmargin}}%
\item\relax}
{\endlist}

\declaretheorem[name=Theorem, numberwithin=section]{theorem}
\declaretheorem[name=Lemma, sibling=theorem]{lemma}

\declaretheorem[name=Corollary, sibling=theorem]{corollary}

\declaretheorem[name=Problem, sibling=theorem]{problem}
\declaretheorem[name=Claim, sibling=theorem]{claim}

\def\cqedsymbol{\ifmmode$\lrcorner$\else{\unskip\nobreak\hfil
\penalty50\hskip1em\null\nobreak\hfil$\lrcorner$
\parfillskip=0pt\finalhyphendemerits=0\endgraf}\fi}

\def\Pr{\mathbb{P}}

\def\E{\mathcal{E}} 

\newcommand{\NN}{\mathbb{N}} 

\tikzstyle{vertex}=[circle, draw, fill=black!50,
                        inner sep=0pt, minimum width=4pt]

\let\leq\leqslant
\let\geq\geqslant

%% file: fig/fig_stars.tex
\begin{tikzpicture}[thick,scale=0.6]
 \node[vertex] (center) at (0,0) {};
 \foreach \pos in {0,1,...,7} {
 \node[vertex] (\pos) at (45*\pos:1.5) {};
 \draw (\pos) -- (center);
 }
 \foreach \pos in {8,9,...,15} {
 \node[vertex] (\pos) at (45*\pos:2.5) { };
 }
 \draw (0) -- (8);
 \draw (1) -- (9);
 \draw (2) -- (10);
 \draw (3) -- (11);
 \draw (4) -- (12);
 \draw (5) -- (13);
 \draw (6) -- (14);
 \draw (7) -- (15);

\begin{scope}[shift={(9,0)}]
 \node[vertex] (center) at (0,0) {};
 \foreach \pos in {0,1,...,7} {
 \node[vertex] (\pos) at (45*\pos:1.5) {};
 \draw (\pos) -- (center);
 }
 \foreach \pos in {8,9,...,15} {
 \node[vertex] (\pos) at (45*\pos:2.5) { };
 }
 \draw (0) -- (8);
 \draw (0) -- (9);
 \draw (2) -- (10);
 \draw (3) -- (11);
 \draw (4) -- (12);
 \draw (5) -- (13);
 \draw (6) -- (14);
 \draw (7) -- (15);

\end{scope}

\end{tikzpicture}

%% file: fig/introlb.tex
\begin{tikzpicture}[level distance=0.75cm,
  level 1/.style={sibling distance=5.25cm},
  level 2/.style={sibling distance=1.75cm}]
  \node {$\varepsilon$}
    child {node {1}
      child {node {11}child {node {$?$} }
        }
      child {node {12} child {node {$?$} }
        }
      child {node {13}child {node {$?$} }
        }
    }
    child {node {2}
        child {node {21} 
            child {node {$?$} }
        }
        child {node {22} child {node {$?$} }
        }
        child {node {23} child {node {$?$} }
        }
    }
    child {node {3}
        child {node {31} child {node {$?$} }
        }
        child {node {32} child {node {$?$} }
        }
        child {node {33} child {node {$?$} }
        }
    };
\end{tikzpicture}

%% file: fig/fig_trees.tex
\begin{tikzpicture}[scale=0.55]
    \node[vertex] [label=above:$v_0$] (u0) at (0,0) { };
    \draw (u0) -- (3,-6) -- ++(-6,0) -- (u0);
    \node[vertex] [label=left:$x$] (x) at (-1,-4) { };
    \draw (x) -- (-2,-6) -- ++(2,0) -- (x);
    \begin{scope}[shift={(-1,-4)}]
        \node[vertex] (s) [label={[label distance=-3.5]0:$s_1$}] at (45:0.5) { };
        \draw (s) -- (x);
    \end{scope}
    \begin{scope}[shift={(s)}]
        \node[vertex] (a) at (90:0.5) { };
        \node[vertex] (b) at (-45:0.5) { };
        \draw (a) -- (s);
        \draw (b) -- (s);
        \node (a) at (0.65,1.65) {$T'$};
    \end{scope}
     \node at (-1,-5.25) {$T_2$};
     \node[vertex] [label=left:$v$] (v) at (-1.5,-6.5) { };
     \draw (v) -- (-1.25,-6);
\end{tikzpicture}

%% file: fig/kchordal_fig.tex
\begin{tikzpicture}[scale = 1.1]
\tikzstyle{vertex}=[circle, draw, fill=black!50,
                        inner sep=0pt, minimum width=4pt]
\draw (3,-4.5) ellipse (5cm and 1.15cm);
\draw (3,1) circle (1cm);
\draw (-1.6,0.7) rectangle (1.6,-3.25);
\begin{scope}[shift={(6,0)}]
    \draw (-1.6,0.7) rectangle (1.6,-3.25);
\end{scope}
\draw (-1.4,-2.75) -- (-1,-4.5) -- (-0.6,-2.75);

\begin{scope}[shift={(0,-2)}]
\draw (7,-2.5) -- (5.5,-0.5) -- (4,-2.5) -- (1,-2.5) -- (0,-0.5) -- (-1,-2.5);
\foreach \x/\y in {0/-0.5, 5.5/-0.5, 4/-2.5, 3/-2.5, 2/-2.5, 1/-2.5, -1/-2.5} {
\node[vertex] at (\x,\y) {};
}
\end{scope}

\draw[snake=coil, segment aspect=0, bend right] (5,0.4) -- (7,-4.5);

\draw (-1,-2.75) circle (0.4cm);
\node at (-1,-2.75) {\scriptsize{$N(u)$}}; 
\node at (3,-1) {$L_{< i-k}$}; 
\node at (3,-1.5) {$L_{\geq i-k}$}; 
\node at (-1.25,-4.75) {$u$}; 
\node at (0,1) {$H$};
\node at (3,2.25) {$S_j$};
\node at (6,1) {$H_x$};
\node at (3,-3.64) {$L_{\geq i}$};

\draw[snake=coil,segment aspect=0] (-1,-2.35) -- (1,0.4);
\node[vertex] (a) at (1,0.4) {};
\draw (a) -- (2.4,1);
\draw[snake=coil,segment aspect=0] (2.4,1) -- (3.6,1);
\draw (3.6,1) -- (5,0.4);

\node[vertex] [label=right:$y$] at (7, -4.5) {};
\node[vertex] [label=right:$x$] at (5,0.4) {};
\node[vertex] (xk) [label=right:$x_{k}$] at (5.95,-1.7) {};

\node[vertex] [label=right:$x_{1}$] at (6.43,-2.9) {};
\node[vertex] [label=right:$x_{2}$] at (6.32,-2.6) {};

\draw[dashed] (-2,-1.25) -- (8,-1.25);

\node at (3,0.6) {$P''$};
\node at (-0.2,-0.6) {$P'$};
\node at (5.9,-0.6) {$P$};

\end{tikzpicture}

%% file: fig/treelowerbound.tex
\begin{tikzpicture}[level distance=1cm,
  level 1/.style={sibling distance=5cm},
  level 2/.style={sibling distance=1.5cm}]
  \node {$\varepsilon$}
    child {node {1}
      child {node {11} child {node {$a_1$} edge from parent}}
      child {node {12} child {node {$a_2$} edge from parent}}
      child {node {13} child {node {$a_5$} edge from parent}}
    }
    child {node {2}
        child {node {21} child {node {$a_3$} edge from parent}}
        child {node {22} child {node {$a_7$} edge from parent}}
        child {node {23} child {node {$a_4$} edge from parent}}
    }
    child {node {3}
        child {node {31} child {node {$a_9$} edge from parent}}
        child {node {32} child {node {$a_6$} edge from parent}}
        child {node {33} child {node {$a_8$} edge from parent}}
    };
\end{tikzpicture}